\documentclass[a4paper,preprintnumbers,floatfix,superscriptaddress,pra,twocolumn,showpacs,notitlepage,longbibliography]{revtex4-1}
\usepackage[utf8]{inputenc}
\usepackage[T1]{fontenc}
\usepackage[sc,osf]{mathpazo}
\usepackage{amsmath, amsthm, amssymb,amsfonts,mathbbol,amstext}
\usepackage{graphicx}
\usepackage{dcolumn}
\usepackage{bm}
\usepackage{bbm}
\usepackage{hyperref}
\usepackage{mathtools}
\usepackage{comment}
\usepackage{color}
\usepackage{multirow}
\usepackage{pgf,tikz}
\usepackage{mathrsfs}
\usetikzlibrary{arrows}

\def\RR{\mathbbm{R}}

\def\1{\mathbf{1}}
\def\0{\mathbf{0}}

\def\st{\textrm{subject to }}

\def\p{\mathbf{p}}

\def\t{\mathbf{t}}

\def\p{\mathbf{p}}

\def\q{\mathbf{q}}

\newcommand{\ket}[1]{| #1 \rangle}

\newcommand{\mean}[1]{\left\langle #1 \right\rangle}

\newtheorem{prop}{Proposition}

\newtheorem{corollary}[prop]{Corollary}
\newtheorem{theorem}[prop]{Theorem}

\renewcommand{\rho}{\varrho}

\newcommand{\processnext}[1]{%
  \ifx\listfinish#1\empty\else\listact{#1}\expandafter\processnext\fi}

{\begin{framed}\begin{small}}
{\end{small}\end{framed}}

\DeclareGraphicsExtensions{.pdf,.png,.jpg}

\makeindex

\begin{document}
\title{Quantifying Bell non-locality with the trace distance}
\date{\today}

\author{S. G. A. Brito}
\affiliation{International Institute of Physics, Federal University of Rio Grande do Norte, 59078-970, P. O. Box 1613, Natal, Brazil}
\author{B. Amaral}
\affiliation{International Institute of Physics, Federal University of Rio Grande do Norte, 59078-970, P. O. Box 1613, Natal, Brazil}
\affiliation{Departamento de F\'isica e Matem\'atica, CAP - Universidade Federal de S\~ao Jo\~ao del-Rei, 36.420-000, Ouro Branco, MG, Brazil} 
\author{R. Chaves}
\affiliation{International Institute of Physics, Federal University of Rio Grande do Norte, 59078-970, P. O. Box 1613, Natal, Brazil}

\begin{abstract}
Measurements performed on distant parts of an entangled quantum state can generate correlations incompatible with classical theories respecting the assumption of local causality. This is the phenomenon known as quantum non-locality that, apart from its fundamental role, can also be put to practical use in applications such as cryptography and distributed computing. Clearly, developing ways of quantifying non-locality is an important primitive in this scenario. Here, we propose to quantify the non-locality of a given probability distribution via its trace distance to the set of classical correlations. We show that this measure is a monotone under the free operations of a resource theory and that furthermore can be computed efficiently with a linear program. We put our framework to use in a variety of relevant Bell scenarios also comparing the trace distance to other standard measures in the literature.

\end{abstract}

\maketitle

\section{Introduction} 
With the establishment of quantum information science, the often called counter-intuitive features of quantum mechanics such as entanglement \cite{Horodecki2009} and non-locality \cite{Brunner2014} have been raised to the status of a physical resource that can be used to enhance our computational and information processing capabilities in a variety of applications. To that aim, it is of utmost importance to devise a resource theory to such quantities, not only allowing for operational interpretations but as well for the precise quantification of resources. Given its ubiquitous importance, the resource theory of entanglement \cite{Horodecki2009} is arguably the most well understood and vastly explored and for this reason has become the paradigmatic model for further developments \cite{Gour2008,Brandao2013,Vicente2014,Winter2016,Coecke2016,Abramsky2017,Amaral2017}.

As discovered by John Bell \cite{Bell1964}, one of the consequences of entanglement is the existence of quantum non-local correlations, that is, correlations obtained by local measurements on distant parts of a quantum system that are incompatible with local hidden variable (LHV) models. In spite of their close connection, it has been realized that entanglement and non-locality refer to truly different resources \cite{Brunner2014}, the most striking demonstration given by the fact that there are entangled states that can only give rise to local correlations \cite{Werner1989}. In view of that and the wide applications of Bell's theorem in quantum information processing, several ways of quantifying non-locality have been proposed \cite{Brunner2014,Popescu1994,Eberhard1993,Toner03,Pironio2003,van2005,Acin2005,Junge2010,Hall2011,Chaves2012,Vicente2014,Fonseca2015,Chaves2015b,Ringbauer2016,Montina2016,Brask2017,GA17}. However, only recently a proper resource theory of non-locality has been developed \cite{Vicente2014,GA17} thus allowing for a formal proof that previously introduced quantities indeed provide good measures of non-local behavior. Importantly, different measures will have different operational meanings and do not necessarily have to agree on the ordering for the amount of non-locality. For instance, a natural way to quantify non-locality is the maximum violation of a Bell inequality allowed by a given quantum state. However, we might also be interested in quantifying the non-locality of a state by the amount of noise (e.g., detection inefficiencies) it can stand before becoming local. Interestingly, these two measures can be inversely related as demonstrated by the fact that in the CHSH scenario \cite{Clauser1969} the resistance against detection inefficiency increases as we decrease the entanglement of the state \cite{Eberhard1993} (also reducing the violation the of CHSH inequality).

Bell's theorem is a statement about the incompatibility of  probabilities predicted by quantum mechanics with those allowed by classical theories. Thus, it seems natural to quantify non-locality using standard measures for the distinguishability between probability distributions, the paradigmatic example being the trace distance. Apart from being a valid distance in the space of probabilities, it also has a clear operational interpretation \cite{Nielsen2010}. However, and somehow surprisingly, apart from the exploratory results in \cite{Bernhard2014}, to our knowledge an in-depth analysis of the trace distance as a quantifier of non-locality has never been presented before.

That is precisely the gap we aim to fill in this paper. In  Sec. \ref{sec:scenario} we describe the scenario of interest and  propose a novel quantifier for non-locality based on the trace distance. Further, in Sec. \ref{sec:LPformulation} we show how our measure can be evaluated efficiently via a linear program and in Sec. \ref{sec:measure} we show that it is a valid quantifier by employing the resource theory presented in \cite{Vicente2014,GA17}. We then apply our framework for a variety of Bell scenarios in Sec. \ref{sec:applications}, including bipartite as well as multipartite ones. In Sec. \ref{sec:relation} we discuss the relation of the trace distance with other measures of non-locality. Finally, in Sec. \ref{sec:discussion} we discuss our findings and point out possible venues for future research.

\section{Scenario}
\label{sec:scenario}
We are interested in the usual Bell scenario setup where a number of distant parts perform different measurements on their shares of a joint physical system. Without loss of generality, here we will restrict our attention to a bipartite scenario (with straightforward generalizations to more parts, see Sec. \ref{sec:applications}) where two parts, Alice and Bob, perform measurements labeled by the random variables $X$ and $Y$ obtaining measurement outcomes described by the variables $A$ and $B$, respectively (see Fig. \ref{fig:bipartite}). 

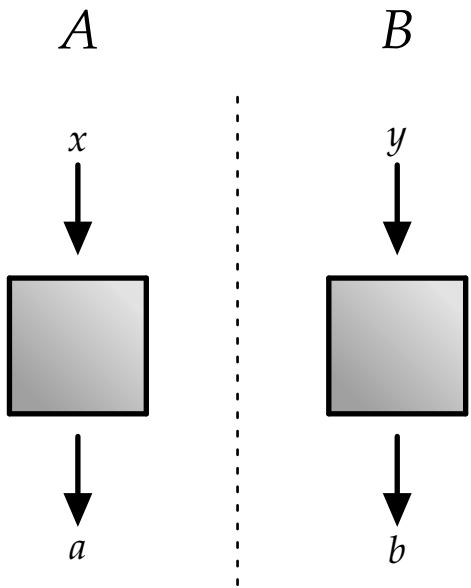
\begin{figure}[h!]
\definecolor{ududff}{rgb}{0.30196078431372547,0.30196078431372547,1}
\definecolor{zzttqq}{rgb}{0,0,0}
\definecolor{aqaqaq}{rgb}{0.6274509803921569,0.6274509803921569,0.6274509803921569}\begin{tikzpicture}[scale=0.3,line cap=round,line join=round,>=triangle 45,x=1cm,y=1cm]
\fill[fill=aqaqaq,fill opacity=1, shading = rectangle, left color=aqaqaq, right color=aqaqaq!30!white, shading angle=135] (-4,8) -- (2,8) -- (2,14) -- (-4,14) -- cycle;
\fill[fill=aqaqaq,fill opacity=1, shading = rectangle, left color=aqaqaq, right color=aqaqaq!30!white, shading angle=135] (10,8) -- (16,8) -- (16,14) -- (10,14) -- cycle;
\draw [line width=2pt,color=zzttqq] (-4,8)-- (2,8);\draw [line width=2pt,color=zzttqq] (2,8)-- (2,14);\draw [line width=2pt,color=zzttqq] (2,14)-- (-4,14);\draw [line width=2pt,color=zzttqq] (-4,14)-- (-4,8);\draw [line width=2pt,color=zzttqq] (10,8)-- (16,8);\draw [line width=2pt,color=zzttqq] (16,8)-- (16,14);\draw [line width=2pt,color=zzttqq] (16,14)-- (10,14);\draw [line width=2pt,color=zzttqq] (10,14)-- (10,8);\draw [->,line width=2pt] (-1,19) -- (-1,15);\draw [->,line width=2pt] (13,19) -- (13,15);\draw [->,line width=2pt] (-1,7) -- (-1,3);\draw [->,line width=2pt] (13,7) -- (13,3);\draw [line width=1pt,dash pattern=on 2pt off 5pt] (6,22)-- (6,0);
\begin{Large}
\draw[color=black] (-1,20) node {$x$};
\draw[color=black] (13,20) node {$y$};
\draw[color=black] (-1,2) node {$a$};
\draw[color=black] (13,2) node {$b$};
\end{Large}
\begin{huge}
\draw[color=black] (-1,25) node {$A$};
\draw[color=black] (13,25) node {$B$};
\end{huge}\end{tikzpicture}
\caption{Bipartite Bell scenario where two parts, Alice  and Bob, share a pair of correlated 
measurement devices, with inputs labeled by $x$ and $y$ and outputs labeled by $a$ and $b$, respectively.}
\label{fig:bipartite}
\end{figure}

A central goal in the study of Bell scenarios is the characterization of what are the distributions $p(a,b \vert x,y)$ compatible with a classical description based on the assumption of local realism implying that
\begin{equation}
\label{LHV}
p_{\mathrm{C}}(a,b \vert x,y)= \sum_{\lambda} p(\lambda) p(a \vert x,\lambda) p(b \vert y,\lambda).
\end{equation}
All the correlations between Alice and Bob are assumed to be mediated by common hidden variable $\lambda$ that thus suffices to compute the probabilities of each of the outcomes, that is, $p(a\vert x,y,b,\lambda)=p(a\vert x,\lambda)$ (and similarly for $b$).

The central realization of Bell's theorem \cite{Bell1964} is the fact that there are quantum correlations obtained by local measurements ($M^x_a$ and $M^y_b$) on distant parts of a joint entangled state $\rho$, that according to quantum theory are described as
\begin{equation}
\label{quantum}
p_{\mathrm{Q}}(a,b \vert x,y)= \mathrm{Tr} \left[ \left(M^x_a \otimes M^y_b \right) \rho \right],
\end{equation}
and cannot be decomposed in the LHV form \eqref{LHV}. Moreover, even more general set of correlations, beyond those achievable by quantum theory and called non-signalling (NS) correlations, can be defined. NS correlations are defined by the linear constraints
\begin{eqnarray}
\label{NS}
& & p(a \vert x) = \sum_{b} p(a,b \vert x,y)= \sum_{b} p(a,b \vert x,y^{\prime}) \\ \nonumber
& & p(b \vert y) = \sum_{a} p(a,b \vert x,y)= \sum_{a} p(a,b \vert x^{\prime},y),
\end{eqnarray}
that is, as expected from their spatial distance, the outcome of a given part is independent of the measurement choice of the other.
The set of classical correlations $\mathcal{C}_{\mathrm{C}}$ (those compatible with \eqref{LHV}) is a strict subset of the quantum correlations $\mathcal{C}_{\mathrm{Q}}$ (compatible with \eqref{quantum}) that in turn is a strict subset of $\mathcal{C}_{\mathrm{NS}}$ (compatible with \eqref{NS}). 

Suppose we are given a probability distribution and want to test if it is non-local or not, that is, whether it admits a LHV decomposition \eqref{LHV}. The most general way of solving that is resorting to a linear program (LP) formulation. First notice that we can represent a probability distribution $q(a,b \vert x,y)$ as a vector $\q$ with a number of components given by $n=\vert x\vert \vert y\vert \vert a\vert \vert b\vert$ ($\vert \cdot \vert$ representing the cardinality of the random variable). 
Thus, \eqref{LHV} can be written succinctly as $\p_{\mathrm{C}}=A \cdot \bm{\lambda}$, with $\bm{\lambda}$ being a vector with components $\lambda_i=p(\lambda=i)$ and A being a matrix indexed by $i$ and $j=(x,y,a,b)$ (a multi-index variable) with $A_{j,i}=\delta_{a,f_a(x,\lambda=i)}\delta_{b,f_b(y,\lambda=i)}$ (where $f_a$ and $f_b$ are deterministic functions). Thus, checking whether $ \q$ is local amounts to a simple feasibility problem that can be written as the following LP:
\begin{eqnarray}
\min_{\lambda \in \RR^m} & & \quad \quad \mathbf{v} \cdot \bm{\lambda}  \\ \nonumber
\st & &  \quad\q=A \cdot \bm{\lambda}  \\ \nonumber
& & \quad \lambda_i \geq 0 \\ \nonumber 
& & \quad \sum_{i}\lambda_i=1,
\end{eqnarray}
where $\mathbf{v}$ represents a arbitrary vector with the same dimension $m=\vert x\vert^{\vert a\vert } \vert y\vert^{\vert b\vert }$ as the vector representing the hidden variable $\bm{\lambda}$.

The measure we propose to quantify the degree of non-locality is based on the trace distance between two probability distributions $\q=q(x)$ and $\p=p(x)$:
\begin{equation}
D(\q,\p)=\frac{1}{2}\sum_{x} \vert q(x) -p(x) \vert. 
\end{equation}
The trace distance is a metric on the space of probabilities since it is symmetric and respect the triangle inequality. Furthermore, it has a clear operational meaning since
\begin{equation}
D(\q,\p)=\max_{S} \vert q(S)-p(S) \vert,
\end{equation}
where the maximization is performed over all subsets $S$ of the index set $\left\{ x\right\}$ \cite{Nielsen2010}. That is, $D(\q,\p)$ specifies how well the distributions can be distinguished if the optimal event $S$ is taken into account. Consider for instance distributions $p(x)$ and $q(x)$ for a variable $X$ assuming $d$ values as $x=1,\dots,d$ and such that $p(x)=\delta_{x,1}$ and $q(x)=\frac{1}{d-1}(1-\delta_{x,1})$. In this case $D(\q,\p)=1$, also meaning that $\q$ and $\p$ can be perfectly distinguished in a single shot since if we observe $x=1$ we can be sure to have $p(x)$ (or $q(x)$ otherwise).

In our case we are interest in quantifying the distance between the probability distribution generated out of a Bell experiment and the closest classical probability (compatible with \eqref{LHV}). We are then interested in the trace distance between $q(a,b,x,y)=q(a,b \vert x,y) p(x,y)$ and $p(a,b,x,y)=p(a,b \vert x,y) p(x,y)$, where $p(x,y)$ is the probability of the inputs and that we choose to fix as the uniform distribution, that is, $p(x,y)=\frac{1}{\vert x \vert \vert y \vert}$. In principle, one could also optimize over $p(x,y)$ and we will do so in Sec. \ref{sec:applications}. However, considering that the measurement choices are totally random and identically distributed is a canonical choice in a Bell experiment.

We are now ready to finally introduce our measure $\mathrm{NL}(\q) $ for the non-locality of distribution $\q=q(a,b \vert x,y)$ that is given by
\begin{eqnarray}
\label{NLtrace}
\mathrm{NL}(\q) & & =\frac{1}{\vert x \vert \vert y \vert } \min_{\p \in \mathcal{C}_{\mathrm{C}}} \quad D(\q,\p)\\ \nonumber
& & =\frac{1}{2\vert x \vert \vert y \vert } \min_{\p \in \mathcal{C}_{\mathrm{C}}} \sum_{a,b,x,y} \vert q(a,b \vert x,y) - p(a,b\vert x,y) \vert.
\end{eqnarray}
This is the minimum trace distance between the distribution under test and the set of local correlations. Geometrically it can be understood (see Fig. \ref{fig:dist}) as how far we are from the local polytope defining the correlations \eqref{LHV}. Therefore, the more we violate a Bell inequality the higher it will be its value. However, as will be further discussed in Sec. \ref{sec:relation}, the violation of a given Bell inequality will in general only provide a lower bound to its value.

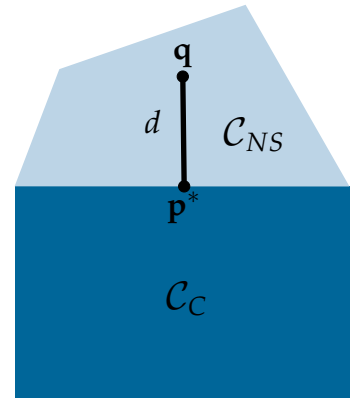
\begin{figure}[h!]
\definecolor{xdxdff}{rgb}{0.49019607843137253,0.49019607843137253,1}
\definecolor{ududff}{rgb}{0.30196078431372547,0.30196078431372547,1}
\definecolor{bcduew}{rgb}{0.7372549019607844,0.8313725490196079,0.9019607843137255}
\definecolor{qqwwzz}{rgb}{0,0.4,0.6}
\begin{tikzpicture}[scale=0.7 ,line cap=round,line join=round,>=triangle 45,x=1cm,y=1cm]
\fill[line width=2pt,color=qqwwzz,fill=qqwwzz,fill opacity=1] (5,1) -- (5,-3) -- (11.292061789564135,-3) -- (11.292061789564135,1) -- cycle;
\fill[line width=2pt,color=bcduew,fill=bcduew,fill opacity=1] (5,1) -- (11.292061789564135,0.9971334792922153) -- (9.34,4.44) -- (5.84,3.22) -- cycle;
\draw [line width=2pt] (8.16,3.08)-- (8.185723211547879,0.9985486567264298);
\begin{large}
\draw [fill=black] (8.16,3.08) circle (3pt);
\draw[color=black] (8.16,3.48) node {$\q$};
\draw [fill=black] (8.185723211547879,0.9985486567264298) circle (3pt);
\draw[color=black] (8.185723211547879,0.5985486567264298) node {$\p^*$};
\draw[color=black] (7.6,2.22702504330775) node {$d$};
\end{large}
\begin{Large}
\draw[color=black] (8.235821423467513,-1.1283849067653924) node {$\mathcal{C}_C$};
\draw[color=black] (9.503240193955191,1.9751661850698143)  node {$\mathcal{C}_{NS}$};
\end{Large}\end{tikzpicture}
\caption{Schematic drawing of a correlation $\q \in \mathcal{C}_{NS}$ and $d=\mathrm{NL}\left(\q\right)$,
the distance (with respect to the $\ell_1$ norm) from $\q$ to the closest local correlation $\p^* \in \mathcal{C}_C$.}
\label{fig:dist}
\end{figure}

\section{Linear program formulation}
\label{sec:LPformulation}

Given a distribution $\q=q(a,b \vert x,y)$ of interest, in order to compute $\mathrm{NL}(\q)$ we have to solve the following optimization problem
\begin{eqnarray}
\min_{\lambda \in \RR^m} & & \quad \| \q - A \cdot \lambda \|_{\ell_1} \\ \nonumber
\st & & \quad \lambda \geq 0 \\ \nonumber 
 & & \quad \sum_{i}\lambda_i=1.
\end{eqnarray}
First notice that the $\ell_1$ norm of a vector $\p$ ($\p \in \RR^n$, with components $p_i$)
\begin{equation}
\|  \p \|_{\ell_1} = \sum_{i} \vert p_i\vert,
\end{equation}
can be written as a minimization problem in the form
 \begin{eqnarray}
\|  \p \|_{\ell_1} = \min_{\t \in \RR^n} & & \quad \langle \1_n, \t \rangle   \label{L1_LP}\\ \nonumber
\st & & \quad -\t \leq \p \leq \t.  
\end{eqnarray}

Using that, we can rewrite our problem as a linear program
\begin{eqnarray}
\label{NLLP1}
\min_{\t \in \RR^n, \lambda \in \RR^m} & & \langle \1_n, \t \rangle   \\ \nonumber
\st & & \quad -\t \leq \q - A \cdot \lambda \leq \t \\ \nonumber
 & & \quad \sum_{i}\lambda_i=1 \\ \nonumber
 & & \quad \lambda \geq 0, \\ \nonumber
\end{eqnarray}
where $\q$ is a known vector of probability distribution to which we want to quantify the non-locality. This way, given an arbitrary distribution of interest we can compute, in an efficient manner, $\mathrm{NL}(\q)$. 

Alternatively, we might be interested not on the full distribution but simply on a linear function of it, for example, the violation of a given Bell inequality in the form $\mathrm{I}_{\mathrm{Bell}} \cdot \q = c $. In this case, further linear constraints need to be added to the LP such as normalization and the fact that the distribution is non-signalling (see Sec. \ref{sec:applications} for examples):
\begin{eqnarray}
\label{NLLP2}
\min_{\t \in \RR^n, \lambda \in \RR^m, \q \in \RR^n} & & \langle \1_n, \t \rangle   \\ \nonumber
\st & & \quad -\t \leq \q - A \cdot \lambda \leq \t \\ \nonumber
 & & \quad \sum_{i}\lambda_i=1 \\ \nonumber
  & & \quad \mathrm{I}_{\mathrm{Bell}}\cdot \q = c \\ \nonumber
 & & \quad \sum_{a,b}\q(ab|xy) = 1 \\ \nonumber
 & & \quad \sum_{a}\q(ab|xy) -\sum_{a}\q(ab|x'y)= 0 \;\;\forall\, (b,y)\\ \nonumber
 & & \quad \sum_{b}\q(ab|xy) -\sum_{b}\q(ab|xy')= 0 \;\;\forall \,(a,x)\\ \nonumber
 & & \quad \lambda \geq 0 \\ \nonumber
 & &\quad \q \geq 0. \\ \nonumber  
\end{eqnarray}
Notice that, instead of adding only NS constraints, one could also be interested in imposing quantum constraints \cite{Navascues2007}. However, in this case we would need to resort to a semi-definite program (that asymptomatically converges to the quantum) instead of a linear program.

Finally, often we might be interested in having an analytical rather than numerical tool. To that aim we can rely on the dual of the LP \eqref{NLLP1} and \eqref{NLLP2}. We refer the reader to \cite{Chaves2015b} for a very detailed account of the dualization procedure but, in short, the optimum solution of \eqref{NLLP1} and \eqref{NLLP2} is achieved in one of the extremal points of the convex set defined by the dual constraints. This way, being able to compute such extremal points $\mathrm{v}_i$, we have an analytical solution, valid for arbitrary test distributions $\q$, given by
\begin{equation}
\mathrm{NL}(\q)=\max_{\mathrm{v}_i} \q \cdot \mathrm{v}_i.
\end{equation}

\section{Proving that trace distance is a non-locality measure}
\label{sec:measure}

A resource theory provides a powerful framework for the formal treatment of a physical property as a resource, enabling its characterization, quantification and manipulation \cite{Gour2008, Brandao2013, Coecke2016}. Such a resource theory  consists  in  three main ingredients:  a set of objects, specifying the physical property that may  serve as a resource, and a characterization of the set of \emph{free objects}, which are the ones that do not contain the resource;  a set of \emph{free operations}, that map every free  object into a free object; resource quantifiers that provide a quantitative characterization of the \emph{amount} of resource  a given object contain. 

One of the essential requirements for a  resource quantifier is that it must be monotonous under free operations, that is, the quantifier 
must not increase when a free operation is applied. Hence, to define proper quantifiers, one needs first to establish
the set of free operations that will be considered, which can vary depending on the applications or the physical constraints under consideration.  

In a resource theory of non-locality, the set of objects is the set of non-signaling correlations $\mathcal{C}_{\mathrm{NC}}$, and the set of free objects is the set $\mathcal{C}_{\mathrm{C}}$ of local correlations. A detailed discussion of several physically relevant free operations for non-locality can be found in Refs. \cite{Vicente2014, GA17}. In what follows we sketch the proof that that our measure $\mathrm{NL}(\q)$ is a monotone for a resource theory of non-locality defined by a wide class of free operations. The detailed proof the results below can be found in the Appendix \ref{sec:proofs}.

The first free operation we consider is relabeling $\mathcal{R}$ of inputs and outputs, defined by a permutation of the set of inputs $x$ and $y$, and outputs $a$ and $b$. This operation corresponds to a permutation of the entries of the correlation vectors $\q$, and hence we have that 
\begin{equation}
\mathrm{NL}\left(\mathcal{R}\left(\q\right)\right) = \mathrm{NL}\left(\q\right).
\end{equation}

Another natural free operation is to take convex sums between a distribution $\q$ and a local distribution $\p \in \mathcal{C}_{\mathrm{C}}$.
In this case, the triangular inequality for the $\ell_1$ norm implies that, if  $\pi \in [0,1]$,
\begin{equation}
\mathrm{NL}\left(\pi \q + \left(1-\pi\right) \p\right) \leq \pi \mathrm{NL}\left( \q \right).
\end{equation}
Combining monotonicity under relabellings and convexity of the $\ell_1$ norm one can show that $\mathrm{NL}$ is monotonous under convex combinations of relabeling operations.

Next, we sketch the proof of the monotonicity of $\mathrm{NL}$ under more sophisticated free operations, namely post-processing and pre-processing operations. Given $\q$, we define a \emph{post-processing operation} as one that transforms $\q$ into $\mathcal{O}\left(\q\right)$,
where
\begin{equation}
\label{eq:post}
\mathcal{O}\left(\q\right) \left(\alpha, \beta \left| x, y\right.\right)=  \sum_{a,b} O^L\left(\alpha, \beta \left|
a,b,x,y\right.\right) \times q\left(a,b\left|x,y\right.\right),
\end{equation}
and $O^L$ is a $x,y$-dependent correlation with inputs $a,b$ and outputs $\alpha, \beta$ that satisfies 
\begin{equation}
O^L\left(\alpha, \beta \left|
a,b,x,y\right.\right) =  \sum_{\lambda} p\left(\lambda\right) O_A^L\left(\alpha \left|
a,x\right.\right) \times O_B^L\left(\beta \left|
b,y\right.\right).
\end{equation}
As shown in Refs. \cite{LVN14,Amaral2017}, output operations preserve the set  of local distributions: if $\q \in \mathcal{C}_{\mathrm{C}}$, then $\mathcal{O}\left(\q\right)\in \mathcal{C}_{\mathrm{C}}$. As a consequence of  convexity of the $\ell_1$ norm, one can prove that 
if $\mathcal{O}$ is an post-processing operation, then 
\begin{equation}
\mathrm{NL}\left(\mathcal{O}\left(\q\right)\right)\leq \mathrm{NL}\left(\q\right).
\end{equation}

Regarding pre-processing operations, it is possible to show that $\mathrm{NL}$ is monotonous under \emph{uncorrelated input enlarging}
operations defined in \cite{Vicente2014}. More generally, one can define a pre-processing operation  that transforms $\q$ into $\mathcal{I}\left(\q\right)$,
where
\begin{equation}
\mathcal{I}\left(\q\right)\left(a,b\left|\chi, \psi\right.\right)= \sum_{x,y} q\left(a,b\left|x,y\right.\right)
I^L\left(x,y\left|\chi , \psi\right. \right),
\label{eq:def_pre}
\end{equation}
and $I^L$ is a local correlation with inputs $\chi, \psi$ and outputs $x, y$.
It was also shown in Refs. \cite{LVN14, Amaral2017} that pre-processing operations preserve the set  of local distributions: if $\q \in \mathcal{C}_{\mathrm{C}}$, then $\mathcal{I}\left(\q\right)\in \mathcal{C}_{\mathrm{C}}$. In what follows, we will consider the 
restricted class of pre-processing operations that satisfy $|x|=\left|\chi\right|$, $|y|=\left|\psi\right|$ and 
\begin{equation}
\sum_{\chi, \psi} I^L\left(x,y\left|\chi , \psi\right. \right) \leq 1.
\label{eq:rest_pre}
\end{equation}
Intuitively, these restrictions forbid one to increase artificially the number of inputs of the scenario and  
\emph{correlated input enlarging} operations, defined in \cite{Vicente2014}. 
As a consequence of this restriction and  convexity of the $\ell_1$ norm, one can prove that 
if $\mathcal{I}$ is an output operation with $|x|=\left|\chi\right|$, $|y|=\left|\psi\right|$ satisfying Eq. \eqref{eq:rest_pre}, then 
\begin{equation}
\mathrm{NL}\left(\mathcal{I}\left(\q\right)\right)\leq \mathrm{NL}\left(\q\right).
\end{equation}


\section{Applications to various Bell scenarios}
\label{sec:applications}

\subsection{Bipartite}
We start considering the paradigmatic CHSH scenario \cite{Clauser1969}, where each of the two parts have two measurement settings with two outcomes each, that is, $x,y,a,b=0,1$. The only Bell inequality (up to symmetries) characterizing this scenario is the CHSH one, that using the notation in \cite{Collins2004} can be written as
\begin{equation}
\mathrm{CHSH}= q_{AB}^{0,0}+q_{AB}^{0,1}+q_{AB}^{1,0} -q_{AB}^{1,1}-q_{A}^{0}-q_{B}^{0} \leq 0,
\end{equation}
where in the inequality above we have used the short-hand notation $q_{A,B}^{x,y}=q(a=0,b=0 \vert x,y)$ and similarly to the other terms. Using the dualization procedure described in Sec. \ref{sec:LPformulation} and assuming $\q$ to be a non-signaling distribution (respecting \eqref{NS}), one can prove that
\begin{equation}
\label{NLchsh}
\mathrm{NL}(\q)=\frac{1}{2}\max \left[0,\Pi(\mathrm{CHSH}) \right],
\end{equation}
where $\Pi(\mathrm{CHSH})$ stand for all the 8 symmetries (under permutation of parts, inputs and outputs) of the CHSH inequality. As expected, in the CHSH scenario the CHSH inequality completely characterizes the trace distance of a given test distribution to the set of local correlations \cite{Masanes2006}. 

Moving beyond the CHSH scenario, we have also considered the CGLMP scenario \cite{Collins2002} where two parts perform two possible measurement with a number $d$ of outcomes. The CGLMP inequality can be succinctly written for any $d$ as:
\begin{eqnarray}
\mathrm{I}^{d}_{\mathrm{CGLMP}} = \frac{1}{4}\sum_{k=0}^{[d/2]-1}\left(1-\frac{2k}{d-1}\right)\nonumber \\
\left[p(a = b + k|00)-p(a=b-k-1|00)\right. +  \nonumber \\
 p(a+k=b|01) - p(a - k - 1= b|01) +\nonumber \\
p(a + k+1 = b|10)- p(a-k=b|10)+  \nonumber \\
p(a=b+k|11)-p(a=b-k-1|11)\left.\right] -\frac{1}{2}\leq 0,
\end{eqnarray}
where $[d/2]$ means the integer part of it and $p(a=b+k|xy) \equiv \sum_{j=0}^{d-1} p(a = j,b = j + k \mod d|xy)$. The maximum value of $\mathrm{I}^{d}_{\mathrm{CGLMP}}$ is $1/2$ and the maximum value for the local variable theories is $0$ $\forall d$ \footnote{We have normalized the inequality in order to obtain the local bound equal $0$}.
In this case we have solved for $d=2,\dots,5$ a LP where instead of fixing the test distribution $q(a,b \vert x,y)$ we only fix the value of the inequality $\mathrm{I}^{d}_{\mathrm{CGLMP}}$ and also impose non-signaling constraints \eqref{NS} over it (see eq. \eqref{NLLP2}). Similarly, to the CHSH case we obtain the same expression given by
\begin{equation}
\mathrm{NL}(\q)=\frac{1}{2}\max \left[0,\mathrm{I}^{d}_{\mathrm{CGLMP}} \right],
\end{equation}
and that we conjecture to hold true to any $d$. Given that the quantum violation of the CGLMP inequality increases with $d$ \cite{Collins2002} it follows that the maximum quantum value of $\mathrm{NL}(\q)$ will follow a similar trend as shown in Fig. \ref{fig:CGLMPplot}.

\begin{figure}[!h]
\centering
\includegraphics[scale=.3]{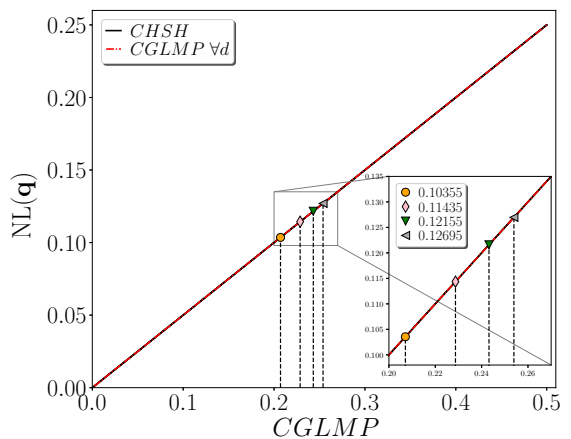}
\caption{Value of $\mathrm{NL}(\q)$ as a function of the value of the CGLMP inequality $\mathrm{I}^{d}_{\mathrm{CGLMP}}$  for $d=2,3,4,5$. The inset shows the maximum known quantum violation for each of these values of $d$ that in turn implies that $\mathrm{NL}(\q)=0.1035$ (d=2), $\mathrm{NL}(\q)=0.1143$ (d=3), $\mathrm{NL}(\q)=0.1215$ (d=4) and $\mathrm{NL}(\q)=0.1269$ (d=5).}
\label{fig:CGLMPplot}
\end{figure}

Finally, another scenario we have considered is the one introduced in \cite{Collins2004}, where each of the two parts measure a number $n$ of observables with 2 outcomes each. We analyze the $\mathrm{I}_{nn22}$ inequality that has the form \cite{Collins2004}
\begin{equation}
\mathrm{I}_{nn22} \leq 0,
\end{equation}
where, for example, for $n=3$ we have
\begin{eqnarray}
\mathrm{I}_{3322}= & &
q_{AB}^{0,0}+q_{AB}^{0,1}+q_{AB}^{0,2}
+q_{AB}^{1,0}+q_{AB}^{1,1}-q_{AB}^{1,2}
\\ \nonumber & & +q_{AB}^{2,0}-q_{AB}^{2,1}
-2q_{A}^{0}-q_{A}^{1}-q_{B}^{0}. 
\label{eq:I3322}
\end{eqnarray}
The maximum non-signaling violation of these inequalities grow linearly with the number of settings $n$ as $\mathrm{I}^{\mathrm{max}}_{nn22}=(n-1)/2$.

Again, by fixing the value of the inequality and imposing the non-signalling constraints we have obtained the corresponding value of $\mathrm{NL}(\q)$, with the results shown in Fig. \ref{fig:Inn22}. Interestingly, even achieving the maximal non-signaling violation of the $\mathrm{I}_{nn22}$ we obtain a value for $\mathrm{NL}(\q)$ that decreases with the number of settings $n$. Therefore, when restricted to this class of inequalities, the best situation is already achieved with $n=2$ (the CHSH scenario). Furthermore this illustrates well the fact that different measures of non-locality do not coincide in general: even though the violation of a Bell inequality can grow with number of setting considered, the trace distance of the corresponding distribution might decrease.

\begin{figure}[h!]
\centering
\includegraphics[scale=.3]{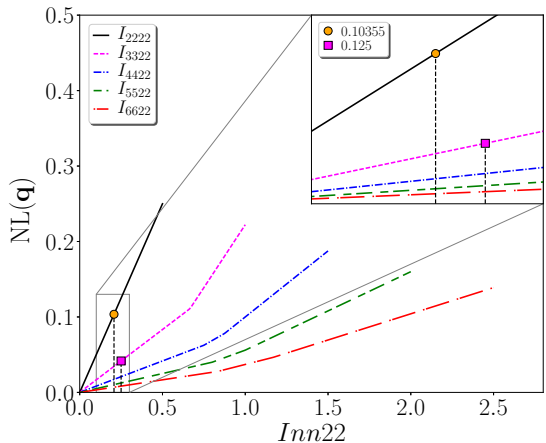}
\caption{Value of $\mathrm{NL}(\q)$ as a function of the value of the  inequality $\mathrm{I}_{nn22}$  for $n=2,3,4,5,6$. The inset shows the maximum known quantum violation (achieved with qubits) for $n=2,3$. Interestingly, even if the maximum NS violation of the inequality increases with the number of settings $n$, its trace distance to the set of local correlations decreases.}
\label{fig:Inn22}
\end{figure}

\subsection{Tripartite}
In the tripartite scenario, considering that each of the parts measure two dichotomic observables, all the different classes of Bell inequalities have been classified \cite{Sliwa2003}. There are 46 of them, under the name of Sliwa inequalities.

Following a similar approach to the CGLMP and $\mathrm{I}_{nn22}$ discussed above, we have computed the value of $\mathrm{NL}(\q)$ as a function of the various Sliwa inequalities. The result in shown in Fig. \ref{fig:sliwa} together with the values associated to the maximum violation of these inequalities. Regarding quantum violations we have used the probability distributions presented in \cite{Rosa2016} where the maximum quantum violation of the Sliwa inequalities has been considered. The results are shown in Fig. \ref{fig:sliwa}b, where we can see that the optimum quantum value of $\mathrm{NL}(\q)=1/8$, higher than the one obtained for the maximum quantum violation of CHSH but smaller than the one for CGLPM already for $d=5$. The maximum non-signalling violation of these inequalities leads to $\mathrm{NL}(\q)=0.25$, the same value obtained for the CGLMP inequality (any $d$) in the bipartite case.

\begin{figure}[tb]
\centering
\includegraphics[scale=.3]{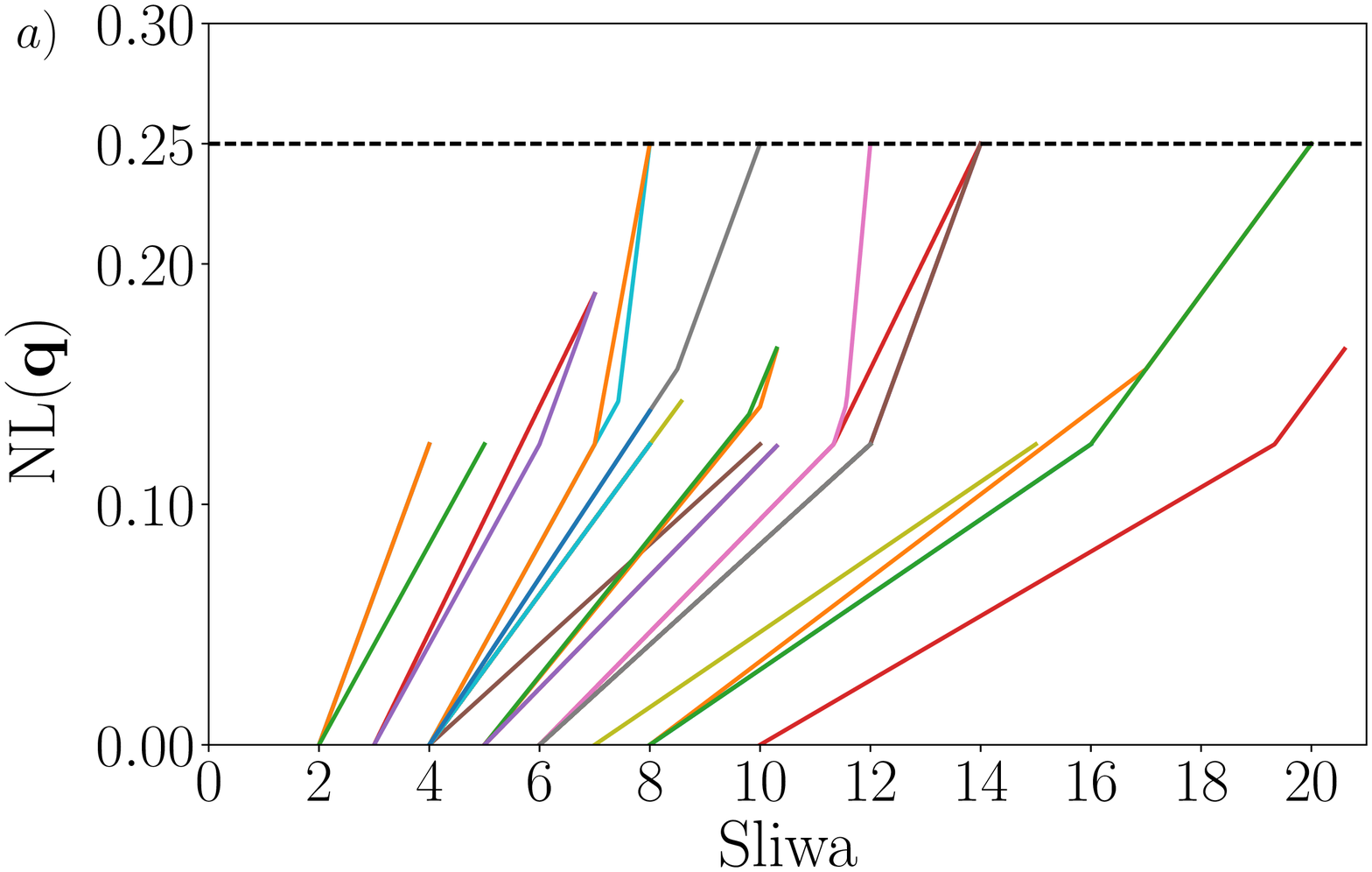}
\includegraphics[scale=.3]{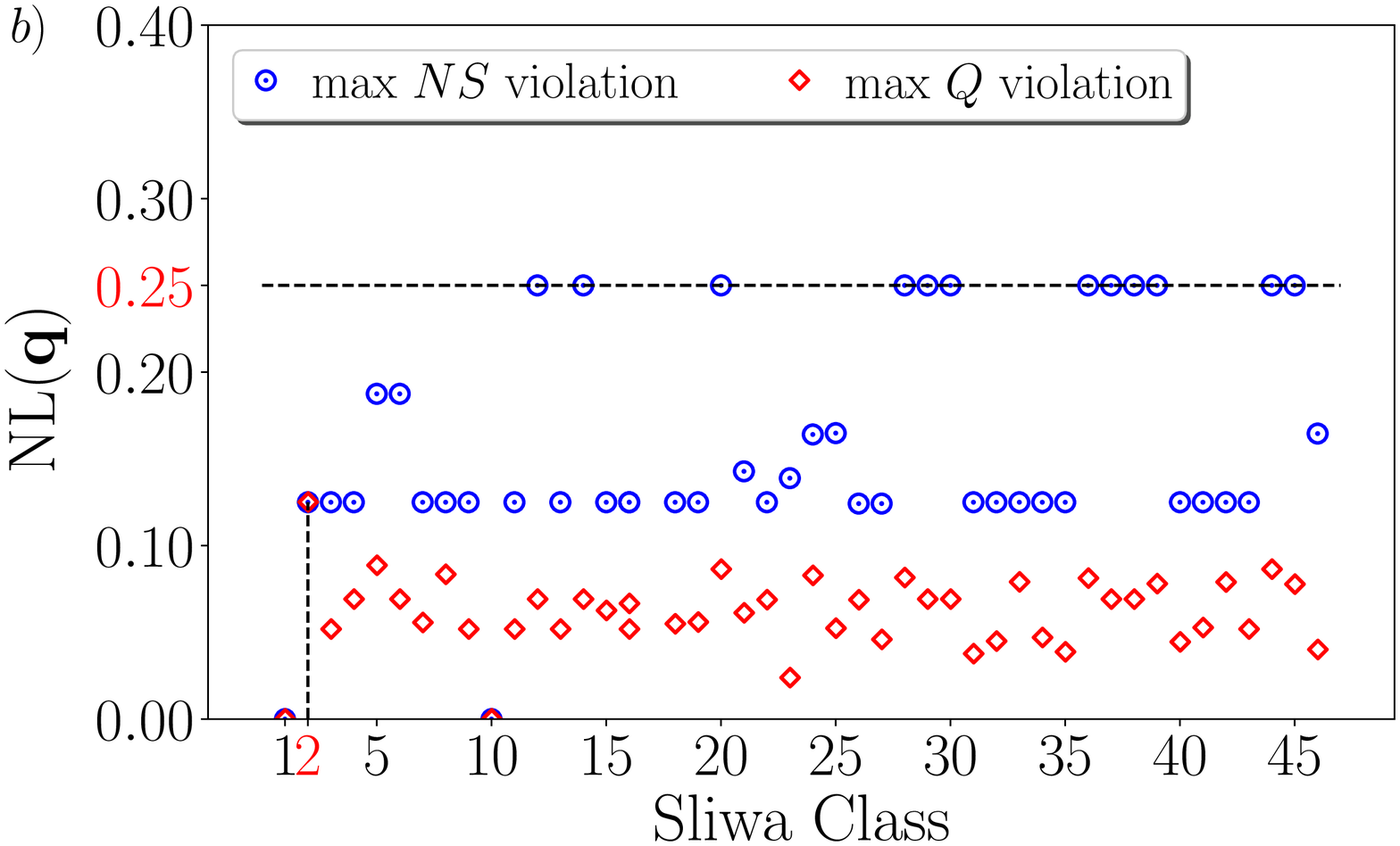}
\caption{a) Values of $\mathrm{NL}(\q)$ as a function of the violation of the $46$ Sliwa classes of inequalities. b) The values of $\mathrm{NL}(\q)$ for the maximal quantum \cite{Rosa2016} and nonsignaling violation of each one of these inequalities. The maximum value of $\mathrm{NL}(\q)=1/8$ is achieved by the maximum violation of the Mermin inequality \cite{Mermin1990} (corresponding to the 2nd Sliwa inequality \cite{Sliwa2003}).}
\label{fig:sliwa}
\end{figure}

\subsection{Mermin inequality for more parts}
The results presented above show that considering the paradigmatic and Sliwa \cite{Sliwa2003} inequalities, the best values we could achieve for our quantifier were $\mathrm{NL}(\q)=1/4$ in the case of non-signalling correlations and $\mathrm{NL}(\q)=1/8$ for quantum correlations. As we show next, these values can be improved analyzing multipartite scenarios beyond 3 parts, more precisely considering the generalization of the Mermin inequality
\cite{Mermin1990,Aderhali1992,Belinskii1993,Belinskii1997} in its form given by \cite{Gisin1998}
\begin{equation}
\langle \mathrm{M}_N \rangle \leq 1,
\end{equation}
that is defined recursively starting with $\mathrm{M}_1=A_1$ by
\begin{equation}
\mathrm{M}_i=\frac{\mathrm{M}_{i-1}}{2}(A_i+\bar{A}_i)+\frac{\bar{\mathrm{M}}_{i-1}}{2}(A_i-\bar{A}_i),
\end{equation}
where $\bar{\mathrm{M}}_{i-1}$ is obtained from $\mathrm{M}_{i-1}$ by exchanging all the observables $A\leftrightarrow \bar{A}$. By choosing suitable projective observables in the $X-Y$ plane of the Bloch sphere and GHZ states \cite{GHZ1989}, the maximum quantum violation is given by $\left\vert \left\langle \mathrm{M}_N\right\rangle \right\vert=2^\frac{N-1}{2}$. For $N$ odd this is also the algebraic/non-signalling maximum of the inequality. For N even the  algebraic/non-signalling maximum is given by $\left\vert \left\langle \mathrm{M}_N\right\rangle \right\vert=2^\frac{N}{2}$. Succinctly, the maximum NS is $\left\vert \left\langle \mathrm{M}_N\right\rangle \right\vert=2^{\lceil \frac{N-1}{2} \rceil}$. 

Following the same approach as before, we have computed the value of $\mathrm{NL}(\q)$ as a function of $M_N$. The results are shown in Fig. \ref{fig:mermin}. Interestingly, we see a clearly increase of $\mathrm{NL}(\q)$ as we consider the maximum violation of $M_N$ with increasing $N$. Notice, however, that going from $N$ even to $N+1$ seems to decrease the value of $\mathrm{NL}(\q)$. The reason for that comes from the fact that from the all possible $2^{N}$ measurement settings allowed by the scenario, only $2^{N-1}$ enter in the evaluation of the $M_N$. Since we are fixing the probability of the inputs to be identically distributed, that amounts to reduce $\mathrm{NL}(\q)$ by a factor of half. If instead, we now choose the probability of inputs to be $1/2^{N-1}$ for all those appearing in $M_N$ and zero otherwise, we then recover a monotonically increasing value of $\mathrm{NL}(\q)$ with $N$. In particular, notice that by doing that we achieve a value of $\mathrm{NL}(\q)=1/4$ for the maximum quantum violation of the Mermin inequality in the tripartite scenario.

\begin{figure}[!htb]
\centering
\includegraphics[scale=.3]{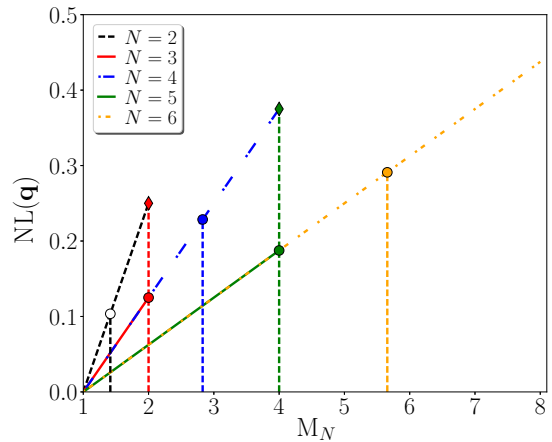}
\caption{$\mathrm{NL}(\q)$ as a function of the values of $\mathrm{M}_N$ for $N=2,3,4,5,6$. We have use the straight line for odd cases and dashed lines for even cases. The circles correspond to $\mathrm{NL}(\q)$ evaluated assuming that each possible measurement setting has a probability $1/2^N$. The diamond correspond to the case where only the measurement settings appearing in $\mathrm{M}_N$ have a probability different of zero (happening for the $N$ odd case).}
\label{fig:mermin}
\end{figure}

In this case, we can also provide an analytical construction, providing an upper bound for $\mathrm{NL}(\q)$, that perfectly coincides with the LP results with $N=2,\dots,6$ and that for this reason we conjecture provides the actual value of $\mathrm{NL}(\q)$ for any $N$. For simplicity in what follows we will restrict our attention to the case of $N$ even.

We choose a fixed distribution $\q=v\q_{max}+(1-v)\frac{1}{2^N}$, where $\q_{max}$ is the distribution given the maximum NS of the inequality. Notice that $M_N$ only contain full correlators and can be succinctly written as
\begin{equation}
\mathrm{M}_N=\frac{1}{2^N}\sum_{x_1,\dots,x_N=0,1} c_{x_1,\dots,x_N}\mean{A^{x_1}_1 \dots A^{x_N}_N}
\end{equation}
with $c_{x_1,\dots,x_N}=\pm 1$. This means that $q_{max}$ is such that $p(a_1,\dots,a_N \vert x_1,\dots,x_N )= 1/2^{N-1} $ if $a_1\oplus \dots \oplus a_N \oplus \delta_{-1,c_{x_1,\dots,x_N}}$ and $p(a_1,\dots,a_N \vert x_1,\dots,x_N)= 0 $ otherwise.

As the probability $\p$ entering in $\| \q - \p \|_{\ell_1}$ we choose
 $\p=v\p_{max}+(1-v)\frac{1}{2^N}$ where $\p_{max}$ is defined by
$p(a_1,\dots,a_N \vert x_1,\dots,x_N )= 1/2^{N-1} $ if $a_1\oplus \dots \oplus a_N \oplus \delta_{-1,c_{x_1,\dots,x_N}}$ (and $0$ otherwise).
We see that for every $c_{x_1,\dots,x_N}=+1$ it follows that $\| \q - \p \|_{\ell_1}=0$. For $c_{x_1,\dots,x_N}=-1$ we have $\| \q - \p \|_{\ell_1}=2v$. So, basically we have to count what is the number of positive and negative coefficients $c_{x_1,\dots,x_N}$ and multiply each of these by $1/2^N$ (assuming all the inputs are equally likely) and by the distance $0$ or $2v$. The number of elements with negative coefficients can be found by a recursive relation. Given that $M_{N}$ had $\alpha_N$ negative coefficients it is easy to see that $\alpha_{N+2}=2\alpha_N+2^{N-2}$ with the initial condition that $\alpha_{2}=1$. So our quantifier for $N$ even is given by
\begin{equation}
\mathrm{NL}(\q)=v(1/2^N)\alpha_{N}=\mathrm{M}_{N}(1/2^{3N/2})\alpha_{N}
\end{equation}
that tends to $\mathrm{NL}(\q)=1/2$ with $N \rightarrow \infty$ and $v=1$ ($\mathrm{M}_N=2^{N/2}$).

\section{Relation to other non-locality measures}
\label{sec:relation}
In the following we will compare the trace distance measure with 3 other standard measures of non-locality: the amount of violation of Bell inequalities, the EPR-2 decomposition \cite{Popescu1994} and the relative entropy \cite{Acin2005}.

\subsection{The violation of Bell inequalities}
The violation of Bell inequalities are a standard way of quantifying the degree of non-locality and in some cases can as well find  operational interpretations, for instance, as the probability of success in some distributed computation protocols \cite{Buhrman2010,Brukner2004}. Clearly, one can expect that the more a given distribution violates a Bell inequality the higher is the trace distance from the local set.

However, in general, the violation of a given Bell inequality only provides a lower bound to such trace distance as can be clearly seen in Fig. \ref{fig:I3322}. In this figure we consider the $\mathrm{I}_{3322}$ inequality (see \eqref{eq:I3322}). Using the LP formulation we have computed $\mathrm{NL}(\q)$ for a distribution of the form $\q= v\q_{\mathrm{max}}+(1-v)1/4$, that is, a mixture of the distribution maximally violating the $\mathrm{I}_{3322}$ with white noise rendering  $\mathrm{I}_{3322}=2v-1$. Clearly, simply imposing the value of the violation of the inequality only gives a non-tight lower bound to $\mathrm{NL}(\q)$. Moreover, we see that optimizing over $p(x,y)$ such that $p(x,y)=1/4 \quad \mathrm{iff} \quad x,y \neq 2$ (in such a way that we recover the CHSH inequality) gives us a higher value for $\mathrm{NL}(\q)$.

\begin{figure}[!htb]
\centering
\includegraphics[scale=.3]{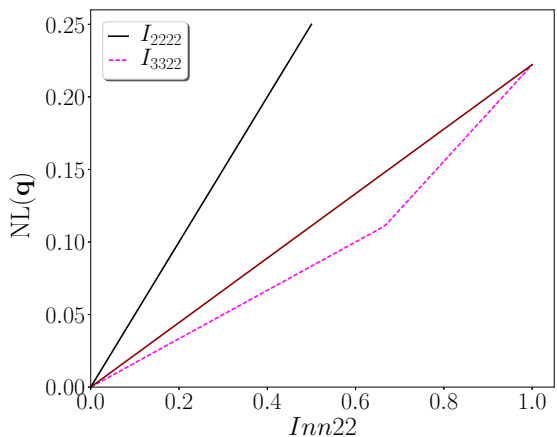}
\caption{The black line shows $\mathrm{NL}(\q)$ for the distribution $\q= v\q_{\mathrm{max}}+(1-v)1/4$ as a function of $\mathrm{I}_{3322}=2v-1$ considering that $p(x,y)=1/4 \quad \mathrm{iff} \quad x,y \neq 2$ (in which case we basically recover $I_{2222}$). The red solid line shows $\mathrm{NL}(\q)$ as a function of $\mathrm{I}_{3322}=2v-1$ considering the full distribution $\q$. Finally, the purple dashed line shows $\mathrm{NL}(\q)$ as a function of $\mathrm{I}_{3322}=2v-1$ considering only the value of the inequality (thus optimizing over all $\q$ compatible with it). Clearly, imposing the value of a given Bell inequality only provides a (in general non-tight) lower bound for $\mathrm{NL}(\q)$.}
\label{fig:I3322}
\end{figure}

\subsection{The EPR-2 decomposition}
Any probability distribution $\q=q(a,b \vert x,y)$ in a Bell experiment can be decomposed into convex mixture of a local part $q_{\mathrm{L}}$ and a non-local and non-signalling part $q_{\mathrm{NL}}$ as $q = (1-w_{\mathrm{NL}}) q_{\mathrm{L}} + w_{\mathrm{NL}}q_{\mathrm{NL}}$, with $0\leq w_{\mathrm{NL}}\leq1$.
The minimum non-local weight over all such decompositions,
\begin{eqnarray}
\label{wtilde}
\tilde{w}_{NL} (\q)&\doteq&\min_{q_{L},q_{NL}}w_{NL}.
\end{eqnarray}
defines the nonlocal content of $\q$, a natural quantifier of the non-locality in $\q$. Nicely, the violation of any Bell inequality $\mathrm{I}\leq\mathrm{I}^L$ (with $\mathrm{I}^L$  the local bound) yields a non-trivial lower bound to $\tilde{w}_{\mathrm{NL}}$ given by
\begin{equation}
\label{lb}
    \tilde{w}_{\mathrm{NL}}(\q)\geq
\frac{\mathrm{I}(\q)-\mathrm{I}^{\mathrm{L}}}{\mathrm{I}^{\mathrm{NL}}-\mathrm{I}^{L}},
\end{equation}
with $\mathrm{I}^{\mathrm{NL}}$ being the maximum value of $\mathrm{I}$ obtainable with non-signaling correlations. Similarly to the trace distance the non-local content can also be computed via a linear program. However, differently from the trace distance, we see that any extremal non-local point of the NS-polytope will achieve the maximum according to this measure, independently of its actual distance to the set of local correlations.

Considering the CHSH scenario and following an identical approach as the one use to get \eqref{NLchsh} one can show that
$
\tilde{w}_{\mathrm{NL}}=2\max \left[0,\Pi(\mathrm{CHSH}) \right]
$. That is, in this particular case we have a simple relation $\tilde{w}_{\mathrm{NL}}(\q)=4\mathrm{NL}(\q)$. This picture, however, changes drastically already at the tripartite scenario. For each of 45 classes of extremal NS points (those violating maximally the Sliwa inequalities) it follows from \eqref{lb} that they achieve the maximum $\tilde{w}_{\mathrm{NL}}(\q)=1$. So, from the perspective of the non-local content all non-local extremal points of the NS polytope display the same amount of non-locality. That is not the case for $\mathrm{NL}(\q)$, as can be clearly seen in Fig. \ref{fig:sliwa} as the different extremal points have different values for it, illustrating the fact that they are at different distances from the local polytope.

\subsection{The relative entropy}
Another important non-locality quantifier is defined in terms of the relative entropy (also known as the Kullback-Leibler divergence) given by \cite{van2005,Acin2005}
\begin{equation}
\label{NLkl}
\mathrm{NL}_{\mathrm{KL}}(\q)=\frac{1}{\vert x \vert \vert y \vert }\min_{\p \in \mathcal{C}_{\mathrm{C}}} \mathrm{KL}(\q,\p),
\end{equation}
with
$
\mathrm{KL}(\q(x),\p(x))= \sum_{x} \q(x) \log{q(x)/p(x)}
$. Similarly to \eqref{NLtrace} we have also chosen an identically distributed distribution for the outcomes. This quantity is also a monotone for the resource theory of non-locality defined by the free operations discussed in Sec. \ref{sec:measure}.
The relevance of this measure comes from the fact of being a standard statistical tool quantifying the average amount of support against the possibility that an apparent non-local distribution can be generated by a local model \cite{van2005}.

\begin{figure}[!htb]
\centering
\includegraphics[scale=.26]{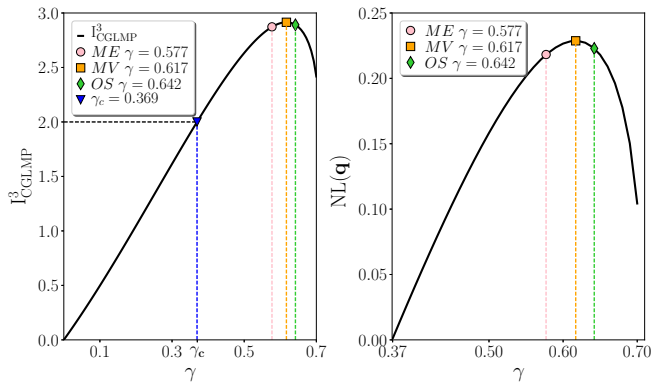}
\caption{In this figure we show the behaviour of the $\mathrm{I}^{3}_{\mathrm{CGLMP}}$ (left) and $\mathrm{NL}(\q)$ (right) inequality as function of $\gamma$ for the distribution $\q$ arising from measurements on the state $\ket{\phi} = \gamma\ket{00}+\sqrt{1-2\gamma^2}\ket{11}+\gamma\ket{22}$. From the critical value of $\gamma_c > 0.369$ we have that $\mathrm{I}^{3}_{\mathrm{CGLMP}}>0$ and $\mathrm{NL}(\q) >0$. In both figures, the circle, square and diamond corresponds to the maximal entangled state $(ME)$, the state that maximal violate (MV) the $\mathrm{I}^{3}_{\mathrm{CGLMP}}$ and the optimal state (OS) that provides the maximal  $\mathrm{KL}(\q,\p)$, respectively.  As we can see, the maximal $\mathrm{I}^{3}_{\mathrm{CGLMP}}$ violation is achieved by $\gamma=0.617$ and the maximal of $\mathrm{NL}(\q)$ is achieved for the same $\gamma$. This is opposed to the result found in \cite{Acin2005}, where the quantum distribution with maximum $\mathrm{KL}(\q,\p)$ do not violated $\mathrm{I}^{3}_{\mathrm{CGLMP}}$ maximally.}
\label{cglmpd3_gamma}
\end{figure}

In \cite{Acin2005} special attention has be given to the relation between $\mathrm{NL}_{\mathrm{KL}}(\q)$ and the violation of the CGLMP inequality \cite{Collins2002}. It has been shown that quantum correlations maximally violating the CGLMP inequality do not necessarily imply the optimum relative entropy. For that, the set of considered quantum distributions comes from a fixed set of measurements on a two-qutrit state of the form $\ket{\phi} = \gamma\ket{00}+\sqrt{1-2\gamma^2}\ket{11}+\gamma\ket{22}$. Inspired by that result we have analyzed what is the value of $\mathrm{NL}(\q)$ in the same setup. The results are shown in Fig. \ref{cglmpd3_gamma} with 3 distributions being specially relevant \cite{Acin2005}: i) the distribution obtained from the maximum violation of CGLPM with maximally entangled states ($\gamma=1\sqrt{3}$), ii) the maximum violation of the CGLMP (obtained with $\gamma=0.617$) and iii) the maximum value of $\mathrm{NL}_{\mathrm{KL}}(\q)$ (achieved with $\gamma=0.642$). As expected from the results in Sec. \ref{sec:applications} and differently for the results obtained for the relative entropy, the more we violate the CGLMP inequality the higher is the trace distance measure $\mathrm{NL}(\q)$.

Interestingly, via the Pinsker inequality \cite{Fedotov2003} we can relate the trace and relative entropy measure. That is, the trace distance provides a non-trivial bound to the relative entropy. Furthermore, the LP solutions to the minimization of the trace distance naturally give an ansatz solution providing an upper bound for $\mathrm{NL}_{\mathrm{KL}}(\q)$. These results are shown in Fig. \ref{cglmpd3} where we plot $\mathrm{NL}_{\mathrm{KL}}(\q)$ as a function of the CGLMP violation. In this plot we also show the curve obtained using the optimal non-signalling distribution $\q^{\ast}$ obtained from the LP used to minimize $\mathrm{NL}(\q)$ subjected to a specific value of the CGLMP inequality (see the caption in Fig. \ref{cglmpd3} for more details).


\begin{figure}[h!]
\centering
\includegraphics[scale=.3]{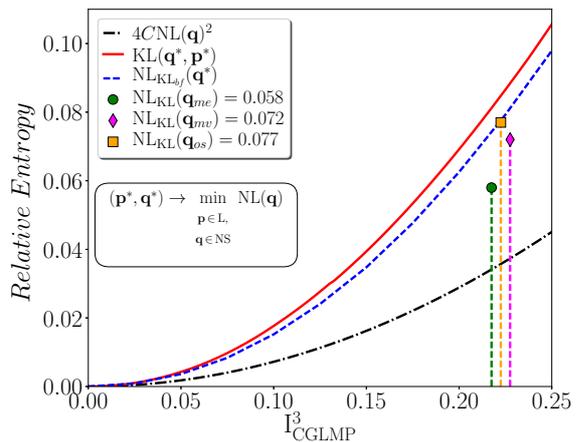}
\caption{Plot of $\min_{\p \in \mathcal{C}_{\mathrm{C}}} \mathrm{KL}(\q,\p)$ as a function of the CGLMP inequality for $d=3$. The black dashed curve represents the lower bound provided the Pinsker inequality, defined by $\mathrm{KL}(\q,\p)\geq 4C{\mathrm{NL}(\q)}^2$, where $C = \frac{1}{2}\log_2{e}$ and $\mathrm{NL}(\q)=\frac{1}{2}\| \q - \p \|_{\ell_1}$. The red curve provides an upper bound for $\min_{\p \in \mathcal{C}_{\mathrm{C}}} \mathrm{KL}(\q,\p)$ since it computes $KL(\q^{\ast},\p^{\ast})$, where $\q^{\ast}$ and $\p^{\ast}$ are the solutions provided by the LP minimization of $\mathrm{NL}(\q)$ subjected to NS constraints and a given value of the CGLMP inequality. The blue curve is obtained by a brute force minimization of $\min_{\p \in \mathcal{C}_{\mathrm{C}}} \mathrm{KL}(\q^{\ast},\p)$. The color points represent the 3 points found in \cite{Acin2005} and discussed in the main text. Interestingly, for the green (circle) and pink (diamond) points there are NS correlations (possibly post-quantum) with the same value of the CGLMP inequality but providing a higher value of $\min_{\p \in \mathcal{C}_{\mathrm{C}}} \mathrm{KL}(\q,\p)$.} 
\label{cglmpd3}
\end{figure}

\section{Discussion}
\label{sec:discussion}
Apart from its primal importance in the foundations of quantum physics, non-locality has also found several applications as a resource in quantum cryptography \cite{Acin2007}, randomness generation/amplification \cite{Pironio2010,Colbeck2012}, self-testing \cite{Mayers2004} and distributed computing \cite{Buhrman2010}. Within these both fundamental and applied contexts, quantifying non-locality is undoubtedly an important primitive.

Here we have introduced a natural quantifier for non-locality, a geometrical measure based on the trace distance between the probability distribution under test and the set of classical distributions compatible with LHV models. Nicely, our quantifier can be efficiently computed (as well as closed form analytical expressions can be derived) via a linear programming formulation. We have shown that it provides a proper quantifier since it is monotonous under a wide class of free operations defined in the resource theory of non-locality \cite{Vicente2014, GA17}. Finally we have applied our framework to a few scenarios of interest as compared our approach to other standard measures.

It would be interesting to use the trace distance measure to analyze previous results in the literature. For instance, the non-locality distillation and activation protocols proposed in \cite{Brunner2011} or the role of the trace distance in cryptographic protocols \cite{Acin2007}. From the statistical perspective one can investigate the relation of the trace distance to the relative entropy \cite{,van2005,Acin2005} and their interconnections to p-values \cite{Elkouss2015} and the statistical significance of Bell tests with limited data \cite{Zhang2011,Bierhorst2015}. Finally, we point out that even though here we have focused on Bell non-locality, the same measure can also be applied to quantify non-classical behavior in more general notions of non-locality \cite{Svetlichny1987,Gallego2012,CAC2017} as well as in quantum contextuality \cite{Abramsky2017,Amaral2017}. We hope our results might motivate further research in these directions.

\begin{acknowledgments}
We thank Daniel Cavalcanti for valuable comments on the manuscript and Thiago O. Maciel for interesting discussions. SGAB, BA and RC acknowledges financial support from the Brazilian ministries MEC and MCTIC. BA also acknowledges financial support from CNPq. 
\end{acknowledgments}

\appendix

\section{Detailed proof of the results in Sec. IV}
\label{sec:proofs}

\subsection{Relabeling operations}
\begin{theorem}
If $\mathcal{R}$ is a relabeling of inputs or outputs of $\q$, then $\mathrm{NL}\left(\mathcal{R}\left(\q\right)\right)=\mathrm{NL}\left(\q\right).$
\end{theorem}
\begin{proof}
Let $\p^* \in \mathcal{C}_C$ be such that 
\begin{equation} \mathrm{NL}\left(\p\right) = \frac{1}{2\left|x\right|\left|y\right|} \left\|\q -\p^* \right\|_{\ell_1}.
\label{eq:q_opt}
\end{equation}
The distributions $\mathcal{R}\left(\q\right)$ and $\mathcal{R}\left(\p^*\right)$ are obtained respectively from $\q$ and $\p^*$ by
a permutation of entries. Then, if follows that 
\begin{equation}
\left\|\mathcal{R}\left(\q\right) -\mathcal{R}\left(\p^*\right) \right\|_{\ell_1}=\left\|\q -\p^* \right\|_{\ell_1}.
\end{equation}
Since relabeling operations preserve the set of local distributions, $\mathcal{R}\left(\p^*\right) \in \mathcal{C}_C$ and hence
\begin{eqnarray}
\mathrm{NL}\left(\mathcal{R}\left(\q\right)\right)& =& \frac{1}{2|x||y|} \min_{\p \in \mathcal{C}_C} \left\|\mathcal{R}\left(\q\right) -\p \right\|_{\ell_1} \\
&\leq & \frac{1}{2|x||y|} \left\|\mathcal{R}\left(\q\right) -\mathcal{R}\left(\p^* \right) \right\|_{\ell_1}\\
&=& \frac{1}{2|x||y|} \left\|\q -\p^*  \right\|_{\ell_1}\\
&=&\mathrm{NL}\left(\q\right)
\end{eqnarray}
Relabeling operations are invertible, and hence there is a relabeling operation $\mathcal{R}^{-1}$ 
such that $\q = \mathcal{R}^{-1}\left[\mathcal{R}\left(\q\right)\right]$. A similar argument shows that
\begin{equation}\mathrm{NL}\left(\q\right) = \mathrm{NL}\left(\mathcal{R}^{-1}\left[\mathcal{R}\left(\q\right)\right]\right) \geq 
\mathrm{NL}\left(\mathcal{R}\left(\q\right)\right)\end{equation}
which proves the desired result.
\end{proof}

\subsection{Convex combinations}
\begin{theorem}
If $\q = \sum_k \pi_k \q_k$, where $\pi_k \leq 0$ and $\sum_k \pi_k=1$, then 
\begin{equation}
\mathrm{NL}\left(\q\right) \leq \sum_k \pi_k \mathrm{NL}\left(\q_k\right).
\end{equation}
\end{theorem}
\begin{proof}
Let $\p_k^* \in \mathcal{C}_C$ be such that
\begin{equation}
\mathrm{NL}\left(\q_k\right) = \frac{1}{2|x||y|}\left\|\q_k -\p_k^* \right\|_{\ell_1}
\end{equation}
and let $\p^*= \sum_k \pi_k\p_k^* \in \mathcal{C}_C$. Then, we have
\begin{eqnarray}
\mathrm{NL}\left( \q\right) & = & \frac{1}{2\left|x\right|\left|y\right|} \min_{\p \in \mathcal{C}_C} \left\|\q-\p \right\|_{\ell_1} \\
&\leq & \frac{1}{2|x||y|}\left\|\q -\p^* \right\|_{\ell_1}\\
&=& \frac{1}{2\left|x\right|\left|y\right|} \left\|\sum_k \pi_k \q_k -\sum_k \pi_k \p_k^* \right\|_{\ell_1}\\
&\leq & \frac{1}{2|x||y|} \sum_k \pi_k \left\| \q_k - \p_k^* \right\|_{\ell_1}\\
&=& \sum_k \pi_k \mathrm{NL}\left(\q_k\right)
\end{eqnarray}
\end{proof}

\begin{corollary}
If $\p \in \mathcal{C}_C$ and $\pi \in [0,1]$,
\begin{equation}
\mathrm{NL}\left(\pi \q + \left(1-\pi\right) \p\right) \leq \pi \mathrm{NL}\left( \q \right).
\end{equation}
\end{corollary}

\subsection{Post-processing operations}

\begin{theorem}
If $\mathcal{O}$ is an post-processing operation, defined as in Eq. \eqref{eq:post}, then $\mathrm{NL}\left(\mathcal{O}\left(\q\right)\right)\leq \mathrm{NL}\left(\q\right).$
\end{theorem}
\begin{proof}
Let $\p^* \in \mathcal{C}_C$ be the distribution satisfying Equation \ref{eq:q_opt}. Then,
\begin{widetext}
\begin{eqnarray}
\mathrm{NL}\left(\mathcal{O}\left(\q \right)\right) &=& \frac{1}{2|x||y|} \min_{\p \in \mathcal{C}_C } \left\|\mathcal{O}\left(\q\right) -\p \right\|_{\ell_1} \\
&\leq & \frac{1}{2|x||y|} \left\|\mathcal{O}\left(\q\right) -\mathcal{O}\left(\p^* \right) \right\|_{\ell_1} \\
& =& \frac{1}{2|x||y|} \sum_{\alpha,\beta,x,y} \left| \sum_{a,b} O^L\left(\alpha, \beta \left|a,b,x,y\right.\right) \left(q\left(a,b\left| x,y\right.\right) -
p^*\left(a,b\left| x,y\right.\right)\right)\right|\\
& \leq & \frac{1}{2|x||y|} \sum_{\alpha,\beta,a,b,x,y}   O^L\left(\alpha, \beta \left| a,b,x,y \right.\right) \left|q\left(a,b\left| x,y\right.\right) - p^*\left(a,b\left| x,y\right.\right)\right|\\
&=&\frac{1}{2|x||y|} \sum_{a,b,x,y} \left|q\left(a,b\left| x,y\right.\right) - p^*\left(a,b\left| x,y\right.\right)\right|\\
&=& \frac{1}{2|x||y|} \left\|\q -\p^* \right\|_{\ell_1}= \mathrm{NL}\left(\q\right).
\end{eqnarray}
\end{widetext}
\end{proof}

\subsection{Pre-processing operations}

In Ref. \cite{Vicente2014} the author defines the uncorrelated input enlarging operation, which consists in one or more parts adding an uncorrelated measurement locally. Without loss of generality we can assume that this measurement is deterministic, since convexity of the $\ell_1$ norm implies that $\mathrm{NL}$ is also a monotone under the addition of a non-deterministic uncorrelated measurement.

Given a correlation $\q$, suppose part $A$ adds one  uncorrelated measurement at her side. Denoting $|x|=m_A$, we define 
\begin{equation}
q_f(a,b|x,y)= \begin{cases}
q(a,b|x,y), & \mbox{if} \ x\leq m_A\\
q(b|y) \delta_{a,a'}, & \mbox{if} \ x=m_A+1 
\end{cases}
\label{eq:input_enl}
\end{equation}
where $a'$ is the deterministic output of the additional measurement $m_A+1$.

\begin{theorem}
If $\q_f$ is obtained from $\q$ by the  input enlarging operation define in Eq, \eqref{eq:input_enl}, then 
\begin{equation}
\mathrm{NL}\left(\q_f\right) \leq \mathrm{NL} \left(\q\right).
\end{equation}
\end{theorem}

\begin{proof}
Let $\p^* \in \mathcal{C}_C$ be the distribution satisfying Equation \ref{eq:q_opt}. For any pair of  inputs $m_A+1, y$
 we have that 
 \begin{widetext}
 \begin{eqnarray}
 \sum_{a,b} \left|q_f\left(a,b|m_A+1,y \right)  - p^*_f\left(a,b|m_A+1,y \right)\right| &=& \\
 & &\sum_{b} \left|q\left(b|y \right)  - p^*\left(b|y \right)\right| \\
& = & \sum_{b} \left|\sum_a q\left(a, b|x, y \right)  - p^*\left(a,b|x, y \right)\right| \ \forall x\leq m_A \\
& \leq  & \sum_{a, b} \left| q\left(a, b|x, y \right)  - p^*\left(a,b|x, y \right)\right| \ \forall x\leq m_A \\
 \end{eqnarray}

 Hence,
 \begin{equation}
 \sum_{a,b} \left|q_f\left(a,b|m_A+1,y \right)  - p^*_f\left(a,b|m_A+1,y \right)\right| \leq \min_{x\leq m_A}  \sum_{a, b} \left|\sum_a q\left(a, b|x, y \right)  - p^*\left(a,b|x, y \right)\right|.
 \end{equation}
  \end{widetext}
  This implies that $\q_f$ is obtained from $\q$ by adding  pairs of inputs  for which the distance of the  probability distributions 
  $q_f$ and $p^*_f$ decreases. Since $\mathrm{NL}$ is defined by taking the average over the pairs of inputs, this implies the
  desired result.
\end{proof}

We now proceed to prove monotonicity under the pre-processing operations defined in Eq. \eqref{eq:def_pre}.

\begin{theorem}
If $\mathcal{I}$ is an pre-processing  operation such that $\left|\chi\right|=|x|$, $\left|\psi\right| =|y|$ and $\sum_{\chi, \psi} I^L\left(x,y\left|\chi , \psi\right. \right) \leq 1$, then 
\begin{equation}
\mathrm{NL}\left(\mathcal{I}\left(\q\right)\right)\leq \mathrm{NL}\left(\q\right).
\end{equation}
\end{theorem}
\begin{proof}
Let $\p^* \in \mathcal{C}_C$ be the distribution satisfying Equation \ref{eq:q_opt}. Then,
\begin{widetext}
\begin{eqnarray}
\mathrm{NL}\left(\mathcal{I}\left(\q \right)\right) &=& \frac{1}{2\left|\chi\right|\left|\psi\right|} \min_{\p \in L} \left\|\mathcal{I}\left(\q\right) -\p \right\|_{\ell_1} \\
&\leq & \frac{1}{2\left|\chi\right|\left|\psi\right|} \left\|\mathcal{I}\left(\q\right) -\mathcal{I}\left(\p^*\right) \right\|_{\ell_1} \\
& =& \frac{1}{2\left|\chi\right|\left|\psi\right|} \sum_{a,b,\chi,\psi} \left| \sum_{x,y}  \left(q\left(a,b\left| x,y\right.\right) -
p^*\left(a,b\left| x,y\right.\right) \right) I^L\left(x,y\left|\chi, \psi \right.\right)\right|\\
& \leq &\frac{1}{2\left|\chi\right|\left|\psi\right|} \sum_{a,b,x,y}  \sum_{\chi, \psi} I^L\left(x,y \left| \chi, \psi \right.\right) \left|q\left(a,b\left| x,y\right.\right) - p^*\left(a,b\left| x,y\right.\right)\right|\\
&=&\frac{1}{2\left|x\right|\left|y\right|} \sum_{a,b,x,y} \left|q\left(a,b\left| x,y\right.\right) - p^*\left(a,b\left| x,y\right.\right)\right|\\
&=& \frac{1}{2\left|x\right|\left|y\right|} \left\|\q -\p^* \right\|_{\ell_1}= \mathrm{NL}\left(\q\right).
\end{eqnarray}
\end{widetext}
\end{proof}

\bibliography{Norm1bib}

\end{document}